\documentclass[aps,prl,reprint,superscriptaddress]{revtex4-2}
\usepackage{amsmath}
\usepackage{amsthm}
\usepackage{mathtools}
\usepackage{amssymb}
\usepackage{mathrsfs}
\usepackage[scr=boondox]
{mathalpha}
\usepackage{graphicx}
\usepackage{dcolumn}
\usepackage{color}
\usepackage{ulem}
\usepackage{bm}
\usepackage{braket}
\usepackage[colorlinks,citecolor=blue]{hyperref}
\usepackage{multirow}
\usepackage{makecell}
\usepackage{booktabs}
\usepackage{graphicx}
\usepackage{booktabs}
\usepackage{multirow}
\usepackage{tabularx}
\usepackage{xcolor}
\usepackage{tikz-cd}
\newtheorem{thm}{Theorem}

\definecolor{c1}{RGB}{104,136,245}
\definecolor{c2}{RGB}{215,112,113}

\newcommand\D{\mathrm{d}}
\newcommand\hrho{\hat{\rho}}
\newcommand\hrhos{\hat{\rho}_\text{ss}}

\newcommand\sL{\mathscr{L}}

\newcommand\sA{\mathscr{a}}
\newcommand\sB{\mathscr{b}}

\newcommand\Le{L_\text{eff}}

\begin{document}

\title{Absence of dissipation-free topological edge states in quadratic open fermions}
\author{Liang Mao}
\email{liangmao.physics@gmail.com}
\affiliation{Institute for Advanced Study, Tsinghua University, Beijing,100084, China}
\affiliation{Institute for Quantum Information and Matter, California Institute of Technology, Pasadena, CA 91125, USA}

\date{\today}

\begin{abstract}
    We prove a no-go theorem: generic quadratic open fermionic systems governed by Lindblad master
equations do not host dissipation-free topological edge states protected by the dissipation gap. By
analogy with topological insulators and superconductors, we map the Lindblad generator to a first-
quantized non-Hermitian matrix representation that plays the role of a band Hamiltonian. Edge
modes of this matrix with vanishing real part are exactly dissipation-free. We show that this matrix
is always adiabatically deformable, through a symmetry-preserving path, to a topologically trivial
Hermitian matrix. Hence no symmetry-protected, dissipation-free edge modes exist in quadratic
open fermions. Our results apply to generic quadratic fermionic Lindbladians and require only
a gapped bulk and a bounded spectrum. They establish a clear boundary for robust topological
phenomena in open fermionic systems.
\end{abstract}

\maketitle

\textit{\color{blue} Introduction.}
Symmetry-protected topological phases of matter, featuring the robust gapless edge states, represents a cornerstone in modern understanding of quantum many-body physics, driving comprehensive theoretical and experimental advances\cite{hasan2010colloquium,qi2011topological,wen2017colloquium}. While traditionally explored in isolated systems, a compelling new frontier has emerged recently: the study of topological phases within open quantum systems\cite{diehl2011topology,bardyn2013topology,ma2023average,ma2025topological,fan2024diagnostics}. Unlike their closed counterparts, open quantum systems are inherently dynamical and typically described by a non-Hermitian quantum master equation\cite{open1,open2}, giving rise to distince topological features\cite{lee2023quantum,bao2023mixed,sang2024mixed,sohal2025noisy,ma2025symmetry,ellison2025toward,wang2025anomaly,lessa2025mixed,guo2025locally,wang2025intrinsic,zemin1,zemin3}.

Within the realm of dissipative topological phenomena, quadratic open fermion systems are arguably the simplest and earliest studied models\cite{3nd1,3nd2,guo2017solutions,yikang,barthel2022solving}. Described by the Lindblad master equations that are quadratic with fermion modes, these systems are the dissipative analogs of well-understood topological insulators and superconductors\cite{kitaev2009periodic,ryu2010topological,wen2012symmetry,chiu2016classification}, thus natural platforms to investigate dissipation-induced topological effects. Initially, these models were explored for their utility in the dissipative preparation of topological states\cite{verstraete2009quantum,budich2015dissipative,iemini2016dissipative,goldman2016topological,goldstein2019dissipation,flynn2021topology}. More recently, out-of-equilibrium topological phenomena are extensively explored, revealing distinctive features such as symmetry-protected transient topological modes\cite{lieu2020tenfold,kawasaki2022topological},
dissipation-driven topological phase of mixed state\cite{tonielli2020topological,altland2021symmetry,mao2024symmetry,mix1,mix2,mix3,mix4}, and novel dissipative dynamics from non-Hermitian skin effect\cite{song2019non,stefano2020unraveling,xue2021simple,haga2021liouvillian,okuma2021quantum}. 

However, despite this exhaustive progress,  the hallmark of topological phases---robust, dissipation-free, symmetry-protected topological edge states---has not been fully realized.
In open systems, a natural generalization of topological edge states require not only protected by the bulk topology, but also robust against dissipation, implying they contribute to the degeneracy of non-equilibrium steady state (NESS). Due to the stringent demand of stability, finding edge states that are both dissipation-free and symmetry-protected remains an elusive challenge.

Here, we provide a negative answer to this problem under general conditions.
We establish a no-go theorem that quadratic open fermionic systems, described by Lindblad master equations 
\begin{align}
	\frac{\D }{\D t}\hrho=\sL[\hrho] =-i[\hat{H},\hrho]+\sum_\mu \big(2\hat{L}_\mu\hrho\hat{L}_\mu^\dagger-\{\hat{L}_\mu^\dagger \hat{L}_\mu,\hrho  \} \big),
\end{align}
 cannot host dissipation-free topological edge states, as long as it has a gapped bulk and bounded spectrum.  Drawing an analogy to topological insulators and superconductors\cite{3nd1,3nd2}, we extract the non-Hermitian matrix representation $\Le$ of $\sL$ which effectively encodes the ``band structure" like their Hermitian counterparts. Any potential topological edge modes of $\Le$ with zero real eigenvalues do not experience dissipation, thus corresponding to the dissipation-free topological edge states. Next, we prove that $\Le$ can always be adiabatically connected to a Hermitian matrix, whose band topology is completely trivial. The adiabatic connection path preserves all the unitary and anti-unitary symmetries of $\Le$, making our results applicable to generic symmetry-protected topology. Thus, we conclude that $\Le$ cannot have zero-energy topological edge modes protected by symmetries, and eliminating the possibility of dissipation-free topological edge states of $\sL$.

\textit{\color{blue} Quadratic Lindbladian and dissipation-free edge states.}
Below we focus on number-conserving complex-fermion
systems for clarity. The Majorana case is discussed later.
Consider  Lindbladians $\sL$ that satisfy the following conditions:
$\hat{H}=\sum_{ij}H_{ij}\hat{c}_i^\dagger \hat{c}_j$ is a quadratic operator of $N$ fermion operators $\hat{c}^\dagger$ and $\hat{c}$. $\hat{L}_\mu$ is linear superpositions of either creation operators $\hat{L}_\mu=\sum_i D^g_{\mu i}\hat{c}_i^\dagger$ or annihilation operators $\hat{L}_\mu=\sum_iD^l_{\mu i}\hat{c}_i$. We also introduce $(M_l)_{ij}=\sum_\mu D^{l*}_{\mu i}D^l_{\mu j}$ and $(M_g)_{ij}=\sum_{\mu}D^{g*}_{\mu i}D^g_{\mu j}$ to characterize the structure of loss and gain. We note both $M_l$ and $M_g$ are positive semi-definite.

$\sL$ plays a similar role for open quantum systems as the many-body Hamiltonian $\hat{H}$ for closed systems. $\sL[\hrho]$ is a linear map from an operator to another operator, often called a \textit{superoperator}. Superoperators are essentially the same as many-body operators on the algebraic level, albeit they act in the operator space. The NESS $\hrhos$ is the eigen-operator of $\sL$ with zero eigenvalue, which is the largest real eigenvalue since other eigen-operator all decay in time. $\hrhos$ plays a role for $\sL$ analogous to that of the many-body ground state for $\hat{H}$.
As pointed out in Refs.~\cite{3nd1,3nd2}, a quadratic $\sL$ has the same structure as quadratic Hamiltonian $\hat{H}$, for it can be written as a bilinear form with a set of fermionic superoperators
\begin{align}
	\sA_i[\hrho]\equiv \hat{c}_i\hrho,&\quad
	\sB_i[\hrho]\equiv \mathcal{P}^F[\hrho]\hat{c}_i,\notag\\
    \sA_i^\dagger[\hrho]\equiv \hat{c}_i^\dagger\hrho,&\quad
	\sB_i^\dagger[\hrho]\equiv -\mathcal{P}^F[\hrho]\hat{c}_i^\dagger
\end{align}
 Here $\mathcal{P}^F[\hrho]=(-1)^{\sum_i\hat{c}_i^\dagger\hat{c}_i}\hrho(-1)^{\sum_i\hat{c}_i^\dagger\hat{c}_i}$ is introduced to restore the fermion anti-commutation relations. $\sA_i$, $\sB_i$, $\sA_i^\dagger$ and $\sB_i^\dagger$ have exactly the same anti-commutation properties as conventional fermion operators.

With $\sA$ and $\sB$, we are able to write $\sL$ as a bilinear form\footnote{$\sL$ has two subspaces, $\sL_\pm$,  defined by $\mathcal{P}^F=\pm 1$. Here what we show is actually $\sL_+$, with matrix representation denoted as $\Le$. On one hand, the $+1$ subspace must contain at least one steady state (but the $-1$ subspace may not). On the other hand, the matrix representation of $-1$ subspace is $\Le'=(\sigma_z\otimes \mathbb{I}_N)\Le(\sigma_z\otimes \mathbb{I}_N)^\dagger$. So $\Le$ and $\Le'$ have the same properties. To maintain the simplicity of discussion, we abuse the terminology and call $\Le$ the matrix representation of the whole Lindbladian $\sL$. In below we will show $\Le$ only have trivial band topology. Thus, the band topology of $\Le'$ is also trivial. These togethor are sufficient to exclude the symmetry-protected edge dark states of $\sL$.   }, 
\begin{align}\label{rep}
	\sL&=
	\begin{pmatrix}
		\sA^\dagger &\sB^\dagger
	\end{pmatrix}
	\begin{pmatrix}
		-iH-M_l+M_g^T & 2M_g^T\\ 2M_l & -iH+M_l-M_g^T
	\end{pmatrix}
	\begin{pmatrix}
		\sA \\\sB
	\end{pmatrix}
	\notag\\
	&\quad+\text{Tr} (iH-M_l-M_g)\notag\\
	&=\sL'+\text{Tr} (iH-M_l-M_g)
\end{align}
Here $\sA$ is a shorthand of column vector $(\sA_1,\cdots,\sA_N)^T$, and the same for $\sB$.
The constant term $\text{Tr}(\cdots)$ guarantees that all the eigenvalues of $\sL$ have non-positive real part. We denote the non-Hermitian matrix representation of the shifted Lindbladian $\sL'$ as $\Le$. 

When $\Le$ is diagonalizable with eigenvalues $\epsilon_i$, $\sL'$ can be written as $\sL'=\sum_{i=1}^{2N} \epsilon_i\overline{\beta_i}\beta_i$. Here $\beta_i$ and $\overline{\beta_i}$ are a set of biorthogonal fermion modes\footnote{Here by biorthogonal we mean $\{\overline{\beta_i},\beta_j\}=\delta_{ij}$, but $\{\beta_i,\beta_j\}$ is not necessarily zero. The is because $\Le$ is a non-Hermitian matrix. In general it can only be diagonalized by similar transformations. However, it does not change the eigenstate structure. It is easy to check that eigenstates of $\sL'$ can also be obtained by acting $\overline{\beta_i}$ on the vacuum.}.
Like its Hermitian counterpart, we now have an operator-vacuum $\hrho_v$ defined by to annihilated by all $\beta_i$s: $\beta_i[\hrho_v]=0$. Eigen-operators of $\sL'$ (and $\sL$)
can be obtained by filling the fermion modes on top of $\hrho_v$ as $\hrho_{i_1,\cdots i_m}=\overline{\beta_{i_1}}[\overline{\beta_{i_2}}[\cdots \overline{\beta_{i_m}}[\hrho_v]\cdots]]$. Remarkably, $\hrhos$, the eigen-operator of $\sL$ with largest real eigenvalue zero, is also the eigen-operator of $\sL'$ with the largest real eigenvalue. Thus, it is obtained by filling all the eigenmodes of $\Le$ whose eigenvalues have a positive real part: 
\begin{align}\label{ness}
    \hrhos=\big(\prod_{\text{Re}\epsilon_i>0}\overline{\beta_i}\big)[\rho_v].
\end{align}

\begin{figure}[!t]
    \centering
    \includegraphics[width=1\linewidth]{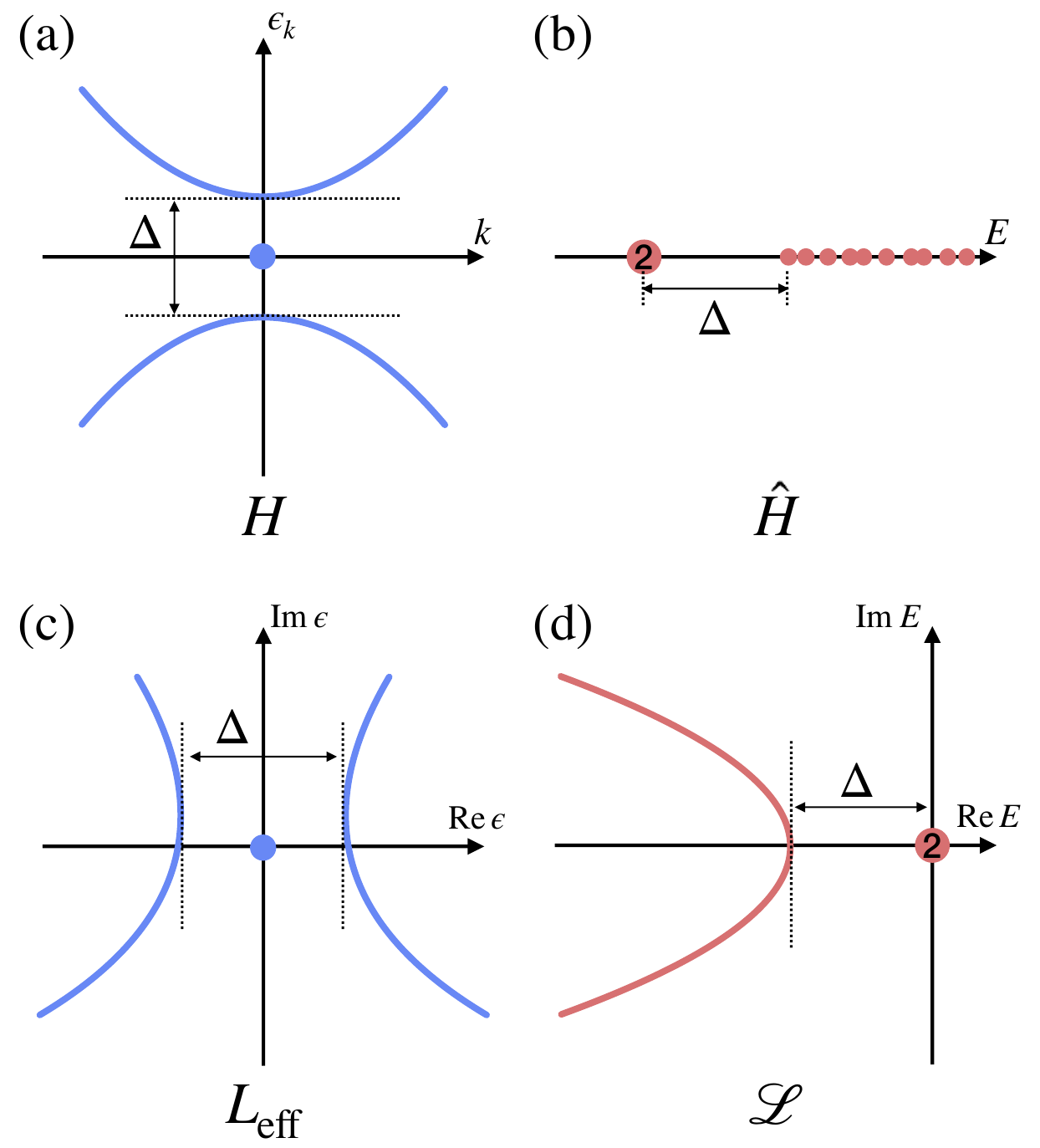}
    \caption{Illustration of energy spectrum, energy gap and edge states of $H$, $\hat{H}$, $\Le$ and $\sL$. We use \textcolor{c1}{\textbullet} and \textcolor{c2}{\textbullet} to distinguish between single-body and many-body spectra, respectively. Number 2 denotes degenerated ground (steady) states.  (a) For closed systems, the topological edge state has zero energy and lies between two bands. This is manifested by the energy spectrum of $H$. (b) Topological edge state of $H$ causes degenerated ground states of $\hat{H}$. (c) Similar to the closed systems, for open systems the dissipation-free edge state has real eigenvalue zero, and lies between two bands. (d) The dissipation-free edge state of $\Le$ causes degenerated steady states of $\sL$. }
    \label{band}
\end{figure}

When $\Le$ possesses an eigenmode $\beta_t$ whose real eigenvalue is zero, $\overline{\beta}_t[\hrhos]$ also acts as the zero eigen-operator of $\sL$, thus is a steady state. This implies that zero-modes of $\Le$ with zero real eigenvalues experience no dissipation. Consequently, any topological edge modes of $\Le$ with zero real eigenvalues, if present, correspond to  robust, dissipation-free topological edge states of $\sL$.
This mirrors the Hermitian setting described by quadratic fermionic Hamiltonians $\hat{H}=\sum_{ij}H_{ij}\hat{c}_i^\dagger \hat{c}_j$, where zero energy edge eigenmode of $H$ contributes to the ground state degeneracy of $\hat{H}$. 
To better illustrate the dissipation-free topological edge states, we compare the single-body and many-body spectra of closed and open systems in Fig.~\ref{band}.

The bulk-edge correspondence principle\cite{bl2,bl3,bl4,bl5} dictates that robust edge modes arise from non-trivial bulk band topology. We therefore work with   periodic boundary condition and assume a gapped bulk for $\Le$ , meaning all eigenvalues $\lambda_i$  satisfy $\operatorname{Re}\lambda_i\in(-\infty,-\Delta/2]\cup[\Delta/2,\infty)$. This ``line-gap" is  naturally inherited from the dissipation gap of $\sL$(Fig.~\ref{band}).
We call bands
with $\operatorname{Re}\epsilon > 0$ “occupied”, since occupying them yields
the NESS via Eq.\eqref{ness}.
 We will  also assume  a bounded spectrum for technical convenience.
Our goal is to show that, under these assumptions, the occupied bands of $\Le$ are topologically trivial, precluding dissipation-free topological edge states.

Before proceeding, we note that non-Hermitian matrices may support robust edge modes not captured by conventional band topology, such as those due to the non-Hermitian skin effect\cite{yao2018edge,kunst2018biorthogonal,okuma2020topological,skineffect,zhang2020correspondence,xiao2020non,NHP}. 
In our setting this subtlety is circumvented: $\Le$, actually any line-gapped non-Hermitian matrices, can be continuously deformed to a Hermitian matrix\cite{bl2,bl3,NHP}, as we make explicit below. Thus, only Hermitian band topology necessitates consideration.

\textit{\color{blue} The quasi-block-diagonalizable structure.}
One crucial property we utilize is the quasi-block-diagonalizable structure of $\Le$.
 Namely, there is a similarity transformation that transforms $\Le$ to a block-diagonal form\cite{wang2024non}
\begin{equation}
\begin{gathered}
S = \frac{1}{\sqrt{2}}
    \begin{pmatrix}
        1 & 1\\ 1+Z& Z-1
    \end{pmatrix},
    S^{-1} = \frac{1}{\sqrt{2}}
    \begin{pmatrix}
        1-Z & 1\\ 1+Z& -1
    \end{pmatrix},\\
   S\Le
    S^{-1}=
    \begin{pmatrix}
        X &0\\0&-X^\dagger
    \end{pmatrix},
\end{gathered}
\label{similar}
\end{equation}
where $X=-iH+M_l+M_g^T$.  $Z$ is a Hermitian matrix and determined by the following equation 
\begin{align}\label{z}
    ZX+X^\dagger Z+2(M_g^T-M_l)=0.
\end{align}
Physically, $Z$ maps to the NESS correlation matrix $C_{ij}=\operatorname{Tr}(\hrhos \hat{c}_i^\dagger\hat{c}_j)$ through $C=(1+Z^T)/2$.

The quasi-block-diagonal structure of $\Le$ suggests that it could be topologicall trivial, since the matrix $X\oplus(-X^\dagger)$ is a trivial matrix. To see this, note that
the Hermitian part of $X$ is a positive semidefinite matrix. As a result, real eigenvalues of $X$ is always non-negative.
When $\Le$ is gapped, real eigenvalues of $X$ are positive. All the occupied bands of $X\oplus(-X^\dagger)$ (those with positive real eigenvalues) together correspond to \textit{all} bands of $X$, thus should be trivial.

We now turn this intuition into a proof in two steps: first, we demonstrate $\Le$ can always be adiabatically connected to a Hermitian matrix $L_H$ with any unitary or antiunitary symmetry constraints.    Then we reveal that $L_H$ exhibits a trivial band topology for all of the occupied bands together.

\textit{\color{blue} Step one: Hermitianization.}
For the first step, we establish the adiabatic path that connects $\Le$ to Hermitian matrix $L_H$.
\begin{thm}\label{thm1}
    There is a continuous path defined by $L(\lambda)$ for $\lambda\in[0,1]$, such that the followings are satisfied:
    \begin{itemize}
        \item $L(0)= L_{\mathrm{eff}}$, and
        \begin{align}
            L(1)=L_H\equiv \begin{pmatrix}
                -Z &1 \\ 1& Z
            \end{pmatrix}.
        \end{align}
        \item \textbf{Adiabaticity}: If $L_{\mathrm{eff}}$ is gapped, $L_\lambda$ is gapped for any $\lambda\in[0,1]$. 
        \item \textbf{Symmetry-preserving}: If $L_{\mathrm{eff}}$ has any unitary or anti-unitary symmetries,  $L(\lambda)$ has the same symmetries for all $\lambda\in[0,1]$.
    \end{itemize}
\end{thm}

The explicit construction of adiabatic path is as follows. For $\lambda\in[0,1/2]$,
\begin{equation}
    L(\lambda)=(1-2\lambda)L_\mathrm{eff}+2\lambda\Big(\oint_{\mathcal{C}_+}-\oint_{\mathcal{C}_-} \Big)     \frac{\mathrm{d} z}{2\pi i}\frac{1}{z-L_\mathrm{eff}}.
\end{equation}
For $\lambda\in[1/2,1]$,
\begin{equation}
L(\lambda)=2(1-\lambda)L\Big(\frac{1}{2}\Big)+\frac{1}{2}(2\lambda-1)\Big(L\Big(\frac{1}{2}\Big)+L^\dagger\Big(\frac{1}{2}\Big)\Big).
\end{equation}
Here the first half path is used for spectrum flattening. The complex contour $\mathcal{C}_+$ is the counter-clockwise contour that encloses all eigenvalues of $X$, which is always well-defined under the gapped and bounded spectrum conditions. $\mathcal{C}_-$ is the symmetric contour of $\mathcal{C}_+$ with imaginary axis.  Then we have
\begin{align}
    L\Big(\frac{1}{2}\Big)&=S^{-1}\Big[\Big(\oint_{\mathcal{C}_+}-\oint_{\mathcal{C}_-} \Big)     \frac{\mathrm{d} z}{2\pi i}
    \bigg(z-
    \begin{pmatrix}
        X & 0 \\0 &-X^\dagger
    \end{pmatrix}
    \bigg)^{-1}\Big]S\notag\\
    &=S^{-1}\begin{pmatrix}
        1 & 0\\ 0 & -1
    \end{pmatrix}
    S
    =
    \begin{pmatrix}
        -Z & 1-Z \\ 1+Z & Z
    \end{pmatrix}.
\end{align}
We emphasize that the above equation holds regardless of whether $X$ is diagonalizable.  Consequently, the adiabatic path remains valid even with the presence of exceptional points\cite{bergholtz2021exceptional,NHP}.

The second half path turns the non-Hermitian matrix $L(1/2)$ to a Hermitian matrix $L(1)$. It is easily verified that $L(1)=\big(L(1/2)+L(1/2)^\dagger\big)/2=L_H$. It remains to prove the adiabaticity and symmetry-preserving of the path. A full proof of Theorem.~\ref{thm1} can be found in SM. See also Ref\cite{NHP}.

Theorem.~\ref{thm1} establishes the topological equivalence between $\Le$ and $L_H$. Specifically, occupied bands of $\Le$ are transformed to positive-energy bands of $L_H$. Topology of $L_H$ is fully understood under the conventional framework for Hermitian matrices.

\textit{\color{blue} Step two: Topological triviality of $L_H$.}
Crucially, $L_H$ has only trivial band topology when considering all the positive bands together. We sketch the proof in below
\begin{thm}\label{thm2}
    Topology of all the positive-energy bands of $L_H$ together are trivial.
\end{thm}

Assume $Z$  has $M$ distinct energy bands. The energy and Bloch wavefunctions are denoted as $\alpha^m_{\vec{k}}$ and $\ket{\psi^m_{\vec{k}}}$, for $1\leq m\leq M$. Then we have
\begin{align}
    L_H&=-\sigma_z\otimes \Big(\sum_{m,\vec{k}}\alpha^m_{\vec{k}}\ket{\psi^m_{\vec{k}}}\bra{\psi^m_{\vec{k}}}\Big)
    +\sigma_x\otimes \Big(\sum_{m,\vec{k}}\ket{\psi^m_{\vec{k}}}\bra{\psi^m_{\vec{k}}}\Big)\notag\\
    &=\sum_{m,\vec{k}} (-\alpha^m_{\vec{k}}\sigma_z+\sigma_x)\otimes \ket{\psi^m_{\vec{k}}}\bra{\psi^m_{\vec{k}}}.
\end{align}
The eigenvalues $E_k$ and Bloch wavefunctions $\ket{\phi_k}$ of $L_H$ are solvable. $L_H$ has $2M$ energy bands. For $1\leq m\leq M$, 
\begin{align}
    E^m_{\vec{k}}=\sqrt{1+(\alpha^m_{\vec{k}})^2},
    \ket{\phi^m_{\vec{k}}}=\ket{v^{+,m}_{\vec{k}}}\otimes \ket{\psi^m_{\vec{k}}}.
\end{align}
For $M+1\leq m\leq 2M$, 
\begin{align}
    E^m_{\vec{k}}=-\sqrt{1+(\alpha^{m-M}_{\vec{k}})^2},
    \ket{\phi^m_{\vec{k}}}=\ket{v^{-,m}_{\vec{k}}}\otimes \ket{\psi^{m-M}_{\vec{k}}}.
\end{align}
Here 
\begin{align}
    \ket{v^{\pm, m}_{\vec{k}}}=\frac{1}{\sqrt{2\big((\alpha^m_{\vec{k}})^2+1\big)}}
    \begin{pmatrix}
        1 \\ \alpha^m_{\vec{k}}\pm\sqrt{(\alpha^m_{\vec{k}})^2+1}
    \end{pmatrix}.
\end{align}
The occupied bands are $\ket{\phi^m_{\vec{k}}}$ for $1\leq m\leq M$.

The tensor product structure of $\ket{\phi^m_{\vec{k}}}$ ensures that we can treat $\ket{v^{+,m}_{\vec{k}}}$ and $\ket{\psi^m_{\vec{k}}}$ separately\cite{kitaev2009periodic,stone2010symmetries} (See SM for a formal justification). For any  $m$,
$\ket{v^{+,m}_{\vec{k}}}$ is a two-dimensional wavefunction resides on the one-dimensional sphere $\mathcal{S}^1$ in the $x-z$ plain, because $\bra{v^{+,m}_{\vec{k}}}\sigma_y\ket{v^{+,m}_{\vec{k}}}=0$.
Consequently, $\ket{v^{+,m}_{\vec{k}}}$ is topologically nontrivial only in one dimension\cite{kitaev2009periodic,ryu2010topological,wen2012symmetry,chiu2016classification}. The topology thereof characterizes the winding around $\mathcal{S}^1$ as $\vec{k}$ traverses the Brillouin zone.
However, 
\begin{align}
\bra{v^{+,m}_{k}}\sigma_x\ket{v^{+,m}_{k}}=\frac{2\big(\alpha^m_{k}+\sqrt{(\alpha^m_{k})^2+1}\big)}
    {\sqrt{2\big((\alpha^m_{k})^2+1\big)}}>0
\end{align}
always holds. As a result, $\ket{v^{+,m}_{k}}$ cannot wrap around $\mathcal{S}^1$.  $\ket{v^{+,m}_{\vec{k}}}$ is  topologically trivial.

Hence, topology of $\ket{\phi^m_{\vec{k}}}=\ket{v^{+,m}_{\vec{k}}}\otimes \ket{\psi^m_{\vec{k}}}$ is solely determined by that of $\ket{\psi^m_{\vec{k}}}$. When summing up $1\leq m\leq M$, we deduce that
\begin{align}\label{principle}
\begin{gathered}
    \text{topology of occupied bands of }L_H\notag\\
    \cong\text{topology of \textit{all} bands of }Z.
    \end{gathered}
\end{align}
For any Hamiltonian with a bounded spectrum, the state that all fermion modes all occupied can be continuously deformed to an atomic insulator. As a result, topology of all bands of $Z$ should be trivial. Therefore the occupied bands of $L_H$ are  also trivial, proving Theorem.~\ref{thm2}.

 Theorem.~\ref{thm1} and \ref{thm2} apply to all the gapped $U(1)-$symmetric quadratic Lindbladians whose spectrum is bounded. 
We proved that these systems must have trivial band topology. Thus, symmetry-protected dissipation-free topological edge states cannot appear.


\textit{\color{blue} Generalization to Majorana fermion systems.}
Our theorems directly extend to quadratic Lindbladians built from Majorana fermions. As evident from the proof, it relies essentially on the existence of similarity transformation \eqref{similar}. Now we demonstrate that for Majorana Lindbladians there exists an analogous  transformation for their matrix representations. Thus subsequent analysis can be performed with the same manner.

For quadratic Lindbladians with Majorana fermions, the system Hamiltonian is the quadratic operator with $2N$ Majorana fermions $\hat{H}=H_{ij}\hat{\gamma}_i\hat{\gamma}_j$ with $H=-H^T$. The dissipation term is $\hat{L}_\mu=l_{\mu i}\hat{\gamma}_i$. Majorana operators satisfy $\{\hat{\gamma}_i,\hat{\gamma}_j\}=2\delta_{ij}$. This type of Lindbladian can be written as a Bogoliubov-de-Gennes form superoperator
\begin{align}
    \sL&=\begin{pmatrix}
		\lambda^\dagger &\lambda
	\end{pmatrix}
	\begin{pmatrix}
		-2iH-M-M^T & 2 (M-M^T)\\
        0 & -2iH+M+M^T
	\end{pmatrix}
	\begin{pmatrix}
		\lambda \\ \lambda^\dagger
	\end{pmatrix}\notag\\
    &=\begin{pmatrix}
		\lambda^\dagger &\lambda
	\end{pmatrix}
    \Le^{\operatorname{BdG}}
    \begin{pmatrix}
		\lambda \\ \lambda^\dagger
	\end{pmatrix},
\end{align}
where we introduce positive-semidefinite Hermitian matrix $M_{ij}=\sum_\mu l_{\mu i}l_{\mu j}^*$. $\lambda$ is the shorthand of $(\lambda_1,\cdots,\lambda_{2N})^T$, with $\lambda_i$ a fermion superoperator defined by $\lambda^\dagger_i[\hrho]\equiv\frac{1}{2}(\hat{\gamma}_i\hrho+\mathcal{P}^F[\hrho]\hat{\gamma}_i)$ and $\lambda_i[\hrho]\equiv\frac{1}{2}(\hat{\gamma}_i\hrho-\mathcal{P}^F[\hrho]\hat{\gamma}_i)$.
Topological zero edge modes of $\Le^{\operatorname{BdG}}$, if exist, correspond to dissipation-free topological edge states of $\sL$.

$\Le^{\operatorname{BdG}}$ admits a similar transformation
\begin{equation}
\begin{gathered}
T = \frac{1}{\sqrt{2}}
    \begin{pmatrix}
        1 & Y\\ 0 & 1
    \end{pmatrix},
    T^{-1} = \frac{1}{\sqrt{2}}
    \begin{pmatrix}
        1 & -Y \\ 0 & 1
    \end{pmatrix},\\
   T\Le^{\operatorname{BdG}}
    T^{-1}=
    \begin{pmatrix}
        -Q^\dagger &0\\0&Q
    \end{pmatrix},
\end{gathered}
\end{equation}
with $Q=-2iH+M+M^T$. $Y$ is determined by
\begin{align}
    YQ^*+Q^TY+2(M-M^T)=0.
\end{align}
Following the same procedure demonstrated previously, one can prove that $\Le^{\operatorname{BdG}}$ is also topologically trivial.

\textit{\color{blue} Summary and discussions.}
In this paper we prove that  fermionic quadratic Lindbladians can only have trivial band topology, thus forbidding the presence of dissipation-free topological edge states. The only assumptions we impose are gapped bulk and bounded spectrum, rendering the most general applicability of our results. In addition to translation-invariant short range systems, our results remain valid for long-range interacting systems, and even Lindbladians with exceptional points. Regarding symmetries, our theorems apply to symmetry-protected band topology by not only antiunitary symmetries but also the unitary ones, including the cases of lattice symmetries such as translation and inversion. This comprehensive approach asserts that any dissipation-free topological edge states protected by these symmetries are not possible. 

Our results build upon, and differentiate from, previous studies concerning topological states in quadratic open fermions. While early research on their dissipative preparation revealed specific models with dissipation-free edge states, their stability under generic perturbation was unclear\cite{diehl2011topology,bardyn2013topology, budich2015dissipative,iemini2016dissipative,goldstein2019dissipation,flynn2021topology}. We demonstrate that these edge states are not robust against general perturbations. Conversely, the damping matrix X itself hosts robust topological edge states\cite{lieu2020tenfold,kawasaki2022topological}, though these only affect transient dynamics and are not dissipation-free. Our results rely on the line-gap of $\Le$. For non-Hermitian matrices, there is another type of gap, the ``point-gap", that is intimately related to the non-Hermitian skin effect. The point-gap topology and skin effect of Lindbladian is also extensively explored recently, in both non-interacting and interacting systems\cite{1,2,3,4,5,6,7,8}.

It is important to note that the conditions of a gapped bulk and bounded spectrum, while quite generic, can be violated in closed systems. For instance, the Weyl semimetal\cite{yan2017topological} famously feature topologically protected gapless bulk spectra. Similarly, the quantum Hall system\cite{von202040} serves as a well-known example of a topological system with an unbounded spectrum. Identifying the dissipative counterparts of these systems may lead to the discovery of new family of topologically nontrivial open fermionic matter. 

While our results apply to non-interacting fermions, interaction can drastically change the picture. Recently, it was shown that steady state of interacting Lindbladians can host topological orders without any symmetry protection\cite{fan2024diagnostics,bao2023mixed,sohal2025noisy,ellison2025toward,wang2025anomaly,lessa2025mixed,wang2025intrinsic}, implying the fundamental difference between non-interacting systems and the interacting ones. Understanding the connections and differences would be a crucial next step toward a more comprehensive understanding of topological phenomena in open quantum systems. We leave these compelling directions for future studies.

\textit{\color{blue} Acknowledgement.}
We are grateful to Hui Zhai, Zhong Wang, Chao-Ming Jian, Fei Song, He-Ran Wang, Zijian Wang, Fan Yang and Tianshu Deng for helpful discussions.

\bibliography{ref}

\onecolumngrid

\clearpage
\subsection*{\large Supplementary Materials for: Absence of dissipation-free topological edge states in quadratic open fermions}

\begin{center}
\noindent {Liang Mao} \\
\vspace{5pt}
\noindent \textit{Institute for Advanced Study, Tsinghua University, Beijing, 100084, China} \\
\noindent \textit{Institute for Quantum Information and Matter, California Institute of Technology, Pasadena, CA 91125, USA}\\
\end{center}
\vspace{10pt}
\normalsize
\setcounter{equation}{0}
\setcounter{figure}{0}
\setcounter{table}{0}
\setcounter{section}{0}
\setcounter{thm}{0}
\setcounter{page}{1}
\renewcommand{\theequation}{S\arabic{equation}}
\renewcommand{\thefigure}{S\arabic{figure}}
\renewcommand{\thetable}{S\arabic{table}}
\renewcommand{\bibnumfmt}[1]{[S#1]}

\twocolumngrid
For the sake of technical rigor, we present and prove formal versions of our theorems.

\subsection{Theorem. 1 and proof}
Before stating the theorem, we clarify the relevant symmetries of non-Hermitian matrices. According to\cite{bl1,bl2,bl3,bl4,bl5}, there are four types of symmetries warrant consideration:
\begin{align}\label{bl}
	\text{K sym:}&\quad \Le=\eta_K U_K\Le^*U_K^{-1},\notag\\
	\text{C sym:}&\quad \Le=\eta_C U_C\Le^TU_C^{-1},\notag\\
	\text{Q sym:}&\quad \Le=\eta_Q U_Q\Le^\dagger U_Q^{-1},\notag\\
	\text{P sym:}&\quad \Le=-U_P\Le U_P^{-1}.
\end{align}
$U_{K,C,Q,P}$ are unitary matrices. $\eta_{K,C,Q}=\pm 1$ denote the sign of the symmetries. The adiabatic path should preserve all the symmetries above.

\begin{thm}
    [formal]
    The following continuous path connects $L(0)=\Le$ and $L(1)=L_H$:
    \begin{itemize}
        \item For $\lambda\in[0,1/2]$,
        \begin{align}
            L(\lambda)=(1-2\lambda)L_\mathrm{eff}+2\lambda\Big(\oint_{\mathcal{C}_+}-\oint_{\mathcal{C}_-} \Big)     \frac{\mathrm{d} z}{2\pi i}\frac{1}{z-L_\mathrm{eff}}.\notag
        \end{align}
        $\mathcal{C}_+$ is the counter-clockwise contour that encloses all the eigenvalues of $X$. $\mathcal{C}_-$ is the symmetric contour of $\mathcal{C}_+$ with imaginary axis.
        \item For $\lambda\in[1/2,1]$,
\begin{equation}
L(\lambda)=2(1-\lambda)L\Big(\frac{1}{2}\Big)+\frac{1}{2}(2\lambda-1)\Big(L\Big(\frac{1}{2}\Big)+L^\dagger\Big(\frac{1}{2}\Big)\Big).\notag
\end{equation}
    \end{itemize}
    Moreover, the path satisifes:
    \begin{itemize}
        \item If all the eigenvalues $\epsilon_i(\Le)$  satisfy $|\operatorname{Re}(\epsilon_i(\Le))|\geq\Delta$, then $|\operatorname{Re}(\epsilon_i(L(\lambda)))|\geq\min\{\sqrt{2}/2,\Delta\}$ for all $\lambda\in[0,1]$ and all $i$.
        \item If $\Le$ has any of the symmetries in Eq.\eqref{bl}, then $L(\lambda)$ has the same symmetries for all $\lambda\in[0,1]$.
    \end{itemize}
\end{thm}
\begin{proof}
    It suffices to prove the two properties of the continuous path for the two half pathes $\lambda\in[0,1/2]$ and $\lambda\in [1/2,1]$.

    \textit{First half path--} Denote $L(1/2) = L_+-L_-$,
    where $L_\pm$ are the (non-Hermitian) projectors onto the eigenspaces of $\Le$ with positive and negative real eigenvalues. Then we have
    \begin{align}
        L(\lambda)=(1-2\lambda)L_\mathrm{eff}+2\lambda(L_+-L_-).
    \end{align}
    Then the eigenvalues satisfy
    \begin{align}
        \epsilon_i(L(\lambda))=(1-2\lambda)\epsilon_i(\Le)+2\lambda \operatorname{sgn}\operatorname{Re} \epsilon_i(\Le).
    \end{align}
    So $|\operatorname{Re}(\epsilon_i(L(\lambda)))|\geq (1-2\lambda)\Delta+2\lambda\geq\min\{1,\Delta\}$.

    Any symmetries can be expressed as 
    \begin{align}
        \Le=\eta U\chi[\Le]U^\dagger,
    \end{align}
    where $\chi[\cdot]$ stands for identity, transpose, complex conjugation or Hermitian conjugation. For all types of $\chi$,
    \begin{align}\label{s-invariant}
        \chi\Big[\oint_{\mathcal{C}_\pm}     \frac{\mathrm{d} z}{2\pi i}\frac{1}{z-L_\mathrm{eff}}\Big]=\oint_{\mathcal{C}_\pm}     \frac{\mathrm{d} z}{2\pi i}\frac{1}{z-\chi[L_\mathrm{eff}]}.
    \end{align}
    This property is obvious for $\chi$ being identity and transpose. For the other two cases, complex conjugation reverses the direction of integration contour to produce a minus sign, which exactly cancels the minus sign from $i\rightarrow -i$. Thus $\chi$ also leaves the whole expression unchanged.

    Applying the symmetry to $L_\pm$, we obtain
    \begin{align}\label{s-inverse}
        \frac{1}{z-\Le}=\frac{1}{z-\eta U\chi[\Le]U^\dagger}=\eta U\frac{1}{\eta z-\chi[\Le]}U^\dagger.
    \end{align}
    Combining Eq.\eqref{s-invariant} and \eqref{s-inverse},
    \begin{align}
        L_+ &= U\Big[\oint_{\mathcal{C}_+}     \frac{\mathrm{d}\eta z}{2\pi i}\frac{1}{\eta z-\chi[L_\mathrm{eff}]}\Big]U^\dagger\notag\\
        &=U\chi\Big[\oint_{\mathcal{C}_+}     \frac{\mathrm{d}\eta z}{2\pi i}\frac{1}{\eta z-\Le}\Big]U^\dagger\notag\\
        &=U\chi[L_\eta]U^\dagger.
    \end{align}
    As a result,
    \begin{align}
        L_+-L_-=\eta U \chi[L_+-L_-]U^\dagger
    \end{align}
    regardless of the sign of $\eta$. This proves the path is symmetry-preserving.

    \textit{Second half path--} The path for $\lambda\in[1/2,1]$ is symmetry-preserving, since $L(1/2)$ has the same symmetries with $\Le$ and $L^\dagger(1/2)$. 

    Expanding the path, we obtain
    \begin{align}
        L(\lambda) = \Big(\frac{3}{2}-\lambda\Big)L\Big(\frac{1}{2}\Big)+
        \Big(\lambda-\frac{1}{2}\Big)L^\dagger\Big(\frac{1}{2}\Big).
    \end{align}
    Since $L(1/2)$ is the minus of two projectors to different subspaces, $(L(1/2))^2=(L^\dagger(1/2))^2=1$. So
    \begin{align}
        (L(\lambda))^2&=\Big[2(\lambda-1)^2+\frac{1}{2}\Big]+\big(\frac{3}{2}-\lambda\Big)
        \Big(\lambda-\frac{1}{2}\Big)\notag\\
        &\quad\quad\cdot\Big[
        L\Big(\frac{1}{2}\Big)L^\dagger\Big(\frac{1}{2}\Big)+
        L^\dagger\Big(\frac{1}{2}\Big)L\Big(\frac{1}{2}\Big)
        \Big]\notag\\
        &\geq \frac{1}{2}.
    \end{align}
    The last inequality results from that the second term is positive semidefinite. As a result, any eigenvalues of $L(\lambda)$ have absolute values greater than $\sqrt{2}/2$.
\end{proof}

\subsection{Theorem. 2 and proof}
We recall that the topology of a quadratic fermion model is formally captured by the vector bundle $E_\phi\rightarrow \mathcal{T}^d$ over $d-$dimensional torus ($d-$dimensional Brillouin zone), whose fiber $F_{\vec{k}}$ is a Hilbert subspace spanned by occupied states. The isomorphic class of $E_\phi$ (under symmetry constraints) forms an element of one of the ten K-groups corresponding to the ten symmetry classes, i.e., $K^{-q}_{\mathbb{C}}(\mathcal{T}^d)$ for $q\in\mathbb{Z}_2$ or $K^{-q}_{\mathbb{R}}(\mathcal{T}^d)$ for $q\in\mathbb{Z}_8$\cite{kitaev2009periodic,ryu2010topological,wen2012symmetry,chiu2016classification,stone2010symmetries}. 
So to claim topological triviality, we need to demonstrate that the isomorphic class of the bundle defined by the occupied states of $L_H$ must be the trivial (identity) element of the K-group, despite symmetry constraints.

\begin{thm}
Let
\begin{align}
P_\phi(k) &= \sum_{m=1}^{M} \ket{\phi_k^{m}}\!\bra{\phi_k^{m}}, \\
\ket{\phi_k^{m}} &= \ket{\nu_k^{+,m}} \otimes \ket{\psi_k^{m}} ,
\end{align}
and assume $k\mapsto P_\phi(k)$ is smooth on $T^d$ and the occupied bands are gapped from the unoccupied bands.
Suppose, for every $m$ and every $k\in T^d$,
\begin{align}
\bra{\nu_k^{+,m}}\sigma_y\ket{\nu_k^{+,m}} &= 0,\\
\bra{\nu_k^{+,m}}\sigma_x\ket{\nu_k^{+,m}} &> 0 .
\end{align}
Then the occupied bundle $E_\phi\to T^d$ is topologically trivial in the appropriate (complex or Real) K-group.
\end{thm}

\begin{proof}
Let $E_\psi^{\,m}$ (resp.\ $E_\nu^{\,m}$) denote the rank-$1$ subbundle spanned by $\ket{\psi_k^{m}}$ (resp.\ $\ket{\nu_k^{+,m}}$). Then
\begin{align}
E_\phi \;\cong\; \bigoplus_{m=1}^{M} \bigl( E_\psi^{\,m} \otimes E_\nu^{\,m} \bigr).
\end{align}
The conditions $\langle\sigma_y\rangle=0$ and $\langle\sigma_x\rangle>0$ imply that $k \mapsto \nu_k^{+,m}$ takes values in the open semicircle $S^1_{x>0}\subset S^1$, which is contractible. Hence  for each $E_\nu^{\,m}$ 
there is a global continuous deformation to trivialize it. In K-theory,
\begin{align}
[E_\nu^{\,m}] = 1 .
\end{align}
Using the ring structure of K-theory ($[E\oplus F]=[E]+[F]$, $[E\otimes F]=[E]\cdot[F]$), we obtain
\begin{align}
[E_\phi]
&= \sum_{m=1}^{M} [E_\psi^{\,m}] \cdot [E_\nu^{\,m}]
 = \sum_{m=1}^{M} [E_\psi^{\,m}]
 = [E_\psi] .
\end{align}
Moreover, $\{\ket{\psi_k^{m}}\}_{m=1}^{M}$ form an orthonormal basis of a fixed subspace $\mathcal H^{\mathrm{dim}(Z)}$ for all $k$, hence
\begin{align}
\sum_{m=1}^{M} \ket{\psi_k^{m}}\!\bra{\psi_k^{m}}
  = \mathbf 1_{\mathcal H^{\mathrm{dim}(Z)}} \qquad \text{for all } k\in T^d ,
\end{align}
which shows $E_\psi \cong T^d \times \mathcal H^{\mathrm{dim}(Z)}$ is the trivial bundle. Therefore $[E_\phi]=1$, as claimed.
\end{proof}

\end{document}